\documentclass[a4paper,UKenglish,cleveref, autoref, thm-restate]{lipics-v2021}

\title{Lexicographic PLS-complete problems: Max-$k$-SAT,  Abelian Permutation Orbit Minimization, and 
approximate pure Nash equilibria} 

\titlerunning{Lexicographic PLS-complete problems} 

\author{Dominik Scheder}{TU Chemnitz, Germany}{dominik.scheder@informatik.tu-chemnitz.de}{0000-0002-9360-7957}{}

\author{Johannes Tantow}{TU Chemnitz, Germany}{johannes.tantow@informatik.tu-chemnitz.de}{0009-0006-0408-6966}{}

\authorrunning{D. Scheder and J. Tantow} 

\Copyright{Dominik Scheder and Johannes Tantow} 



\category{} 

\relatedversion{} 

\acknowledgements{}

\nolinenumbers 

\usepackage[caps, full]{complexity}
\usepackage{tikz}
\usetikzlibrary{arrows.meta, positioning}
\newclass{\CLS}{CLS}
\newlang{\flip}{\textup{\textsc{Flip}}}
\newlang{\DCR}{DCR}
\newlang{\sat}{SAT}

\newcommand{\Z}{\mathbb{Z}}

\newcommand{\N}{\mathbb{N}}

\newcommand{\reals}{\mathbb{R}}
\newcommand{\NAND}{\textnormal{NAND}}
\newcommand{\nth}[1]{#1^{\textsuperscript{th}}}
\newcommand{\localmin}{\textsc{Polm}} 
\newcommand{\group}[1]{\langle #1 \rangle}

\newcommand{\successors}{\textnormal{Succs}}

\newcommand{\Value}{\textnormal{Value}}
\newcommand{\startsolution}{\textnormal{StartSolution}}
\newcommand{\feasible}{\textnormal{Feasible}}

\newcommand{\localmaxcut}{\textup{\textsc{LocalMaxCut}}}
\newcommand{\localmaxksat}[1]{\textup{\textsc{LocalMax}-#1-\textsc{SAT}}}
\newcommand{\lexmaxksat}[1]{\textup{\textsc{LocalLexMax}-#1-\textsc{SAT}}}
\newcommand{\polm}{\textup{\textsc{Polm}}} 

\newtheorem{searchproblem}[theorem]{Search Problem}

\numberwithin{equation}{section}

\usepackage{crossreftools}
\hideLIPIcs  

\EventEditors{John Q. Open and Joan R. Access}
\EventNoEds{2}
\EventLongTitle{42nd Conference on Very Important Topics (CVIT 2016)}
\EventShortTitle{CVIT 2016}
\EventAcronym{CVIT}
\EventYear{2016}
\EventDate{December 24--27, 2016}
\EventLocation{Little Whinging, United Kingdom}
\EventLogo{}
\SeriesVolume{42}
\ArticleNo{23}

\begin{document}

\maketitle

\begin{abstract}
How hard is it to find a local optimum? If we are given a graph and want 
to find a locally maximal cut--meaning that the number of edges in the cut cannot be improved by
moving a single vertex from one side to the other--then just iterating improving steps finds a local maximum since the size 
of the cut can increase at most $|E|$ times. If, on the other hand, the edges are weighted, this 
problem becomes hard for the class PLS (Polynomial Local Search)~\cite{DBLP:journals/siamcomp/SchafferY91}.

We are interested in optimization problems with {\em lexicographic costs}. For Max-Cut this 
would mean that the edges $e_1,\dots, e_m$ have costs $c(e_i) = 2^i$. For such a cost 
function, it is easy to see that finding a {\em global} Max-Cut is easy. In contrast, we show that 
it is PLS-complete to find an assignment for a 4-CNF formula that is locally maximal (when the clauses 
have lexicographic weights); and also for a 3-CNF when we relax the notion of ``local'' by allowing 
to switch two variables at a time. 

This result is interesting for two different reasons. Firstly, it offers another starting point for reductions to problems that are lexicographic by definition.
We use these results to answer a question in Scheder and Tantow~\cite{SchederTantow2025}, who showed that 
finding a lexicographic local minimum of a string $s \in \{0,1\}^n$ 
under the action of a list of given permutations $\pi_1, \dots, \pi_k \in S_{n}$ is PLS-complete. They ask whether the 
problem stays PLS-complete when the $\pi_1,\dots,\pi_k$ commute, i.e., generate an Abelian subgroup $G$
of $S_n$.  We show that it does, and in fact stays PLS-complete even (1) when every element in $G$ has order two and also (2) 
when $G$ is cyclic, i.e., all $\pi_1,\dots,\pi_k$ are powers of a single permutation $\pi$.

Secondly, we use it to further investigate the complexity of computing pure $\alpha$-Nash equilibria in congestion games. Using lexicographic 4-SAT, we obtain a simple proof of the PLS-completeness originally shown by Skopalik and Vöcking\cite{DBLP:conf/stoc/SkopalikV08}. The simple structure allows us to also show hardness for exponential and polynomial delay functions with positive coefficients. Additionally, every resource is used by a constant amount of players and every player has a constant number of strategies. However, the degree of the polynomials is not bounded by a constant but grows linearly with the number of clauses and logarithmically with the approximation ratio $\alpha$.
\end{abstract}

\newpage

\section{Introduction}

    Polynomial Local Search (\PLS) is the complexity class of problems that require us to find a local optimum.     
    Formally, a \PLS{} problem is defined by a set $\mathcal{I}$ of instances, 
    a set $\mathcal{S}$ of solutions, a function $\feasible: \mathcal{I} \times \mathcal{S} \rightarrow \{\texttt{true},\texttt{false}\}$
    telling us whether $s$ is at all a feasible solution to instance $x$,
    and a function $\startsolution: \mathcal{I} \rightarrow \mathcal{S}$ where $\startsolution(x)$ is always a 
    feasible solution for $s$; further, 
    functions 
    $\Value: \mathcal{I} \times \mathcal{S} \rightarrow \N$, where $\Value(x,v)$ computes how good the solution $v$ is for 
    the instance $x$; and a {\em successor function} $\successors: \mathcal{I} \times \mathcal{S} \rightarrow \mathcal{S}^*$,
    where $\successors(x,v) = (w_1,\dots,w_k)$ means that $w_1,\dots, w_k$ are the neighboring solutions to solution $v$. 
    We require that all feasible solutions to $x$ have the same length, and that all functions are 
    polynomially computable. In particular, this means that each solution $v$ has at most polynomially many successors
    $w_i$. A feasible solution $v$ is {\em locally optimal} if 
    none of its successors $w_i$ is {\em better}, i.e., $\Value(x,v) \geq \max_i \Value(x,w_i)$ (for maximization problems) 
    or $\Value(x,v) \leq \min_i \Value(x,w_i)$ (for minimization problems). The goal is to find a locally optimal feasible 
    solution.
    
    The study of this class was initiated by Johnson, Papadimitriou and Yannakakis in \cite{DBLP:journals/jcss/JohnsonPY88}. In their work they showed that 
    \flip{} is \PLS{}-complete:
    
    \begin{searchproblem}[\flip{}]
    Given a circuit $C$ with $n$ input and $m$ output bits; find an input $x \in \{0,1\}^n$ 
    such that $C(x) \succeq_{\rm lex} C(x \oplus \mathbf{e}_i)$ for all $i = 1,\dots, n$. Here, $\succeq_{\rm lex}$ is 
    the lexicographical order on $\{0,1\}^m$, i.e., we interpret the outputs as binary numbers
    and compare them accordingly.
    \end{searchproblem}

    The instances of \flip{} are circuits; feasible solutions are bit vectors of the appropriate length; the function $\Value$ is given by the output of the circuit and the neighborhood of a bit vector are all bit vectors with a Hamming distance of 1.
    
    In the realm of search problems, the notion of a {\em reduction from $A$ to $B$} 
    requires a function $f$ mapping $A$-instances to $B$-instances and a function $g$ mapping $B$-solutions to 
    $A$-solutions. Several search problems are known to be \PLS{}-complete under 
    these reductions:
        
    \begin{enumerate}
    \item \localmaxksat{$k$} for $k \geq 2$. Given a $k$-CNF formula $F = C_0 \wedge \dots \wedge C_{m-1}$ with $n$ variables and 
    clause weights $w_1, \dots, w_m$, 
    find an assignment $\alpha$ such that the weight of satisfied clauses cannot be improved by flipping a single variable (Krentel~\cite{DBLP:conf/focs/Krentel89}, Schäffer and Yannakakis~\cite{DBLP:journals/siamcomp/SchafferY91}).
    \item \localmaxcut{}. Given a graph $G =(V,E)$ and edge weights $w: E \rightarrow \reals$, 
    find a cut $V = S \uplus T$
    such that its weight (the sum of weights of all edges crossing the cut) cannot be improved by moving a vertex from 
    one side to the other 
    (Schäffer and Yannakakis~\cite{DBLP:journals/siamcomp/SchafferY91}).

    \item \textsc{Permutation Orbit Local Minimum} ($\polm$): Given a string $s \in \{0,1\}^n$ and 
    permutations $\pi_1,\dots, \pi_k \in S_n$, find $\pi \in \group{\pi_1,\dots,\pi_k}$ such that 
    $s \circ \pi$ is lexicographically locally minimal, i.e, $s \circ \pi \preceq_{\rm lex} s \circ \pi \circ \pi_i$ for all $1 \leq i \leq k$
    (Scheder and Tantow~\cite{SchederTantow2025}).
    \item Finding pure Nash equilibria in congestion games, described in more detail in \cref{section-congestion-games} below (Fabrikant, Papadimitriou and Talwar~\cite{DBLP:conf/stoc/FabrikantPT04}, Ackermann, Röglin and Vöcking\cite{DBLP:conf/focs/AckermannRV06}).
     \end{enumerate}

    The first two are obviously in \PLS. \polm{} is because membership $\pi \in \group{\pi_1,\dots,\pi_k}$ can be decided efficiently (Furst, Hopcroft and Luks~\cite{DBLP:conf/focs/FurstHL80}). For congestion games, the existence of pure Nash equilibria can be shown using 
    an appropriate potential function that decreases whenever a player decreases their cost by changing their strategy~(Rosenthal~\cite{rosenthal1973class}), which then also shows membership in \PLS.

    As described by Giannakopoulos, Grosz, and Melissourgos~\cite{DBLP:conf/icalp/GiannakopoulosG24}, the function $\Value{}$ can be described as a function mapping a solution $s$ to some solution vector $s'$ and the result is then $c \cdot s'$
    for some cost vector $c$. In \localmaxksat{$k$}, $s'$ in position $i$ has a $1$ if the clause $C_i$ is satisfied by $s$ and $0$ else. The cost vector $c$ contains then the weight each clause.
    In both \flip{} and \polm{}, the value of a solution is just a bit string, interpreted as a binary number. Therefore, we have that $s'$ is the bit string and $c_i = 2^i$. We say they use 
    a {\em lexicographic cost function}. This motivates the following problem:

    \begin{searchproblem}[\lexmaxksat{$k$}]
     \lexmaxksat{$k$} is \localmaxksat{$k$} where the clause $C_i$ has a weight of $2^i$. 
     Equivalently, we can view the problem as a FLIP instance with $n$ input bits and $m$ output bits where every output bit is computed by a single $k$-clause.
    \end{searchproblem}     
    
    We are now ready to state our first main result:

    \begin{theorem}
    \label{theorem-lex-4-sat-pls-complete}
    \lexmaxksat{$4$} is \PLS-complete, and so is \flip{} with depth-1 circuits. 
    \end{theorem}

    We will prove Theorem~\ref{theorem-lex-4-sat-pls-complete} in Section~\ref{section-lex-4-sat}. Note that 
    \lexmaxksat{$2$} can be solved polynomially; in fact, we can even compute a global maximum by the following 
    greedy algorithm: 
    
    \begin{enumerate}
        \item Set $G_m := 1$, the empty CNF, always true.
        \item For $i = m-1, \dots, 0$: set $G_i := G_{i+1} \wedge C_i$ if this is satisfiable; otherwise, set $G_i := G_{i+1}$.
    \end{enumerate}
    This algorithm works for all $k$-CNF formulas as long as the clause weights are lexicographic; it is efficient as long 
    as we can check efficiently whether $G_{i+1} \wedge C_i$ is satisfiable. Similarly, \localmaxcut{} with lexicographic costs is in \P{}.\\

    One of the main ways to show the \PLS{}-hardness of a problem $L$ 
    is a reduction from FLIP using Krentel's technique\cite{DBLP:conf/coco/Krentel89, DBLP:conf/focs/Krentel89} of simulating multiple circuits,  
    one on the current input and some on (a subset of) 
    its neighborhood and using a gadget to switch to the best neighbor. 
    One must then use the machinery of $L$ to simulate how the circuits evaluate their 
    inputs and how one switches from the current solution to a better neighbor. 
    This is used in the proof of \cref{theorem-lex-4-sat-pls-complete} as well as in \cite{SchederTantow2025, DBLP:journals/tcs/DumraufM13, DBLP:conf/stoc/FabrikantPT04, DBLP:journals/siamcomp/SchafferY91, DBLP:conf/stoc/SkopalikV08, DBLP:conf/icalp/FotakisKLMPS20}. These proofs are often quite long and complicated. 
    
    \lexmaxksat{4} is a special case of \flip{} that is still \PLS{}-complete and combines the simple cost function of \flip{} with the simple combinatorial structure of CNF formulas. We argue that
    for reductions showing PLS-completeness, \lexmaxksat{4} can often
    be a better starting point than \flip{}, leading to simpler proofs and stronger results. 
    We demonstrate this in two cases: 

    \begin{theorem}
    \label{theorem-abelian-polm}
    \textup{\textsc{Permutation Orbit Local Minimum}} is \PLS-complete even when restricted to permutations 
    $\pi_1,\dots,\pi_k$ that generate an Abelian group. 
    \end{theorem}

    We prove \cref{theorem-abelian-polm} in Section~\ref{section-abelian-polm}; in fact, we prove that it stays \PLS-complete 
    even if we require that all permutations in 
    $\group{\pi_1,\dots,\pi_k}$ have order 2. In \cref{section-cyclic-polm} we prove that it stays \PLS-complete even 
    if we require the group $\group{\pi_1,\dots,\pi_k}$ to be cyclic. This answers an open question and strengthens the result from \cite{SchederTantow2025} as Abelian and cyclic groups allow less freedom for the permutations and also admit a potentially faster brute-force algorithm as $\pi \circ \sigma = \sigma \circ \pi$.\\
    
    Our final application of Theorem~\ref{theorem-lex-4-sat-pls-complete} comes from algorithmic game theory:

    \begin{theorem}
        \label{theorem-game}
        Finding an $\alpha$-equilibrium for $\alpha > 1$ in a congestion game is PLS-complete when all delay 
        functions are step functions or all are exponential functions or all are polynomials with positive coefficients.
    \end{theorem}    
    Skopalik and Vöcking have shown that finding an $\alpha$-equilibrium in congestion games 
    with step functions is \PLS-complete\cite{DBLP:conf/stoc/SkopalikV08}. Using Theorem~\ref{theorem-lex-4-sat-pls-complete}, we can give a
    very short and simple proof of this fact. In the appendix of the work of Caragiannis, Fanelli, Gravin and Skopalik \cite{DBLP:conf/focs/CaragiannisFGS11} this was also shown for polynomials with positive coefficients. Nonetheless, the question of the complexity of finding approximate Nash equilibria with positive coefficient polynomials was raised in \cite{DBLP:conf/icalp/FotakisKLMPS20,DBLP:conf/wine/VijayalakshmiS20, waldmann2022congestion}.
    
    Using \lexmaxksat{4} we get an extremely simple reduction for both cases that is also very regular and easy to analyze. In the light of this hardness result it is natural to ask whether there are useful parameters of the game that allow us to solve the problem efficiently when these parameters are small.
    
    In all reductions the maximum number of players per resource is bounded (3 in \cite{DBLP:conf/stoc/SkopalikV08}, 2 in \cite{DBLP:conf/focs/CaragiannisFGS11} and 4 in our work).
    Another interesting parameter is the maximum number of strategies per player. In \cite{DBLP:conf/focs/CaragiannisFGS11,DBLP:conf/stoc/SkopalikV08} this grows linearly with the input due to the $\textnormal{check}^j_{i,b}$ strategies of the $Y_j$ players. Our reduction
    shows that the problem is hard even when each player has only two strategies.\footnote{Players with one strategy can be removed by changing the payoff functions, so this is the minimum number of strategies players can have in any (non-trivial) game.}
    We will prove Theorem~\ref{theorem-game} in Section~\ref{section-congestion-games}.

    All reductions in this work are tight, i.e. the standard algorithm that starts from a solution and moves to a better neighbor takes for any problem in the worst-case exponentially many steps and deciding if a solution is the result of the standard algorithm is \PSPACE{}-complete. (see \cref{section-tightness}).

\section{Computing approximate Nash equilibria in congestion games - proof of Theorem \ref{theorem-game}}
    \label{section-congestion-games}
     In a congestion game, there is a set $N = \{1, \dots, n\}$ of players and a set $E$ of resources. Each player $i \in N$ has 
     a set of {\em strategies} $S_i \subseteq 2^E$. Each resource $e \in E$ has a {\em delay} $d_e : \mathbb{N} \to \mathbb{N}$.
    A {\em strategy profile} is a vector $\mathbf{s} = (s_1, s_2, \dots, s_n)$ with $s_i \in S_i$. The {\em load} of a resource $e$ 
    under profile $\mathbf{s}$ is the number of players using it: 
    $l(e,\mathbf{s}) := |\{i \in N \mid e \in s_i\}|$. This load causes a  delay of $d_e (l(e,\mathbf{s}))$. The cost of such a strategy profile $\mathbf{s}$ for a player $i$ is the sum of the delays of each resource they use, i.e. $C_i(\mathbf{s}) = \sum_{e \in s_i} d_e(l(e,\mathbf{s}))$. Players act selfishly and switch their strategy if it reduces the cost. A pure Nash equilibrium is a strategy profile $\mathbf{s}=(s_1,\dots,s_n)$ where no player has 
    an incentive to unilaterally change their strategy, i.e., if 
    \begin{align*}
    \forall i \in N\  \forall s_i' \in S_i: 
    C_i(\mathbf{s}) \leq C_i(s_1,\dots,s_{i-1}, s_i', s_{i+1}, \dots, s_n)  \ . 
    \end{align*}
    For such games, a pure Nash equilibrium always exists, as shown by Rosenthal~\cite{rosenthal1973class}. 
    Fabrikant, Papadimitriou, and Talwar~\cite{DBLP:conf/stoc/FabrikantPT04} show that finding
    a pure Nash equilibrium is \PLS-complete by giving a simple reduction from 
    local NAE-$3$-SAT.
    When we assume that the players are lazy and only change their strategy if their cost decreases by a significant amount, we get the concept of an 
    {\em approximate} pure Nash equilibrium. More formally, for a real $\alpha \geq 1$, an $\alpha$-equilibrium is a strategy profile $\mathbf{s}= (s_1,\dots,s_n)$ such that 
    \begin{align*}
    \forall i \in N\  \forall s' \in S_i: 
    C_i(\mathbf{s}) \leq \alpha \cdot C_i(s_1,\dots,s_{i-1}, s_i', s_{i+1}, \dots, s_n)  \ . 
    \end{align*}
    Skopalik and Vöcking~\cite{DBLP:conf/stoc/SkopalikV08} show that this is \PLS-complete for any computable 
    $\alpha > 1$ using an involved reduction from \flip{} and even more extended in \cite{DBLP:conf/focs/CaragiannisFGS11} for polynomials with positive coefficients.

    It is not possible to reuse the general reduction from $3$-SAT to show hardness of approximation. Consider the clauses $A = \neg y$, $C = x$, $C' = \neg x$ and $D = x \lor y$ with the weights $w(A) = 8$ , $w(C) = w(C') = 2$ and $w(D) = 1$. The assignment $x, y \mapsto 0,0$ has a weight of $w(C) + w(D) = 3$, whereas flipping the value of x leads to an assignment of weight $w'(C) = 2$. Thus, this is an improving move, but not if it should improve by an $\alpha \geq 2$ and similar examples exist for larger $\alpha$. 
    These problems don't occur with lexicographic weights, and using a reduction 
    from $\lexmaxksat{4}$, we can give a very short proof.

    \begin{proposition}[Skopalik and Vöcking~\cite{DBLP:conf/stoc/SkopalikV08}]
        For any computable $\alpha > 1$ and any congestion game $G$, finding an $\alpha$-equilibrium is \PLS-complete.
    \end{proposition}
    \begin{proof}
        We reduce from \lexmaxksat{$4$}. Let $F = C_0 \wedge C_1 \wedge \dots \wedge C_{m-1}$
        be a weighted 4-CNF formula with $w(C_j) = 2^{j}$.
        For an assignment $\alpha$, define
        $C(\alpha) := (C_{m-1}(\alpha), \dots, C_0(\alpha)) \in \{0,1\}^m$. The total weight of satisfied clauses 
        is thus simply $C(\alpha)$, interpreted as a binary number. Thus, $\Value(\alpha) \leq \Value(\beta)$
        if and only if $C(\alpha) \preceq_{\rm lex} C(\beta)$, i.e., only the first differing bit matters. This means 
        we can change the clause weights from $2^{j}$ to $c^{j}$ for any $c \geq 2$ and all 
        the answers to $\Value(\alpha) \stackrel{?}{\leq} \Value(\beta)$ will stay the same.
        We choose $c = \alpha + 1$.
        
        A solution to \lexmaxksat{$4$} is an assignment that locally maximizes the weight 
        of satisfied clauses. In the following proof it might be better to think about 
        {\em minimizing} the weight of {\em unsatisfied} clauses (which is of course the same).
        
        We will now define the congestion game. 
        The resources are the clauses $C_0, \dots, C_{m-1}$. The players are the variables 
        $\{x_1,\dots,x_n\}$. Player $x_i$
        has two strategies:
        \begin{align*}
            \textnormal{true}_i & := \{C \in F \mid \neg x \in C\} \\ 
            \textnormal{false}_i & := \{C \in F \mid x \in C\} . 
        \end{align*}        
        Note that strategy profiles in the game are in a 1-to-1 correspondence 
        with the Boolean assignments to the variables of $F$. A player (i.e., variable $x$) 
        uses the resources (i.e., clauses) that are {\em not} satisfied by the current assignment
        and contain the variable $x$.
        The delay $d_C$ of clause/resource $C_j$ is defined to be $0$ if 
        strictly fewer than $k$ players use it (where $k \leq 4$ is the size of $C_j$), 
        and $w(C_j) = (1+\alpha)^j$ if $k$ players do (note that $C_j$ 
        can never be used by more than $k$ players). The more valuable clauses thus cause a higher delay for the player then less valuable clauses. Apart from the value 
        $1 + \alpha$, this is roughly as in~\cite{DBLP:conf/stoc/FabrikantPT04}.

        We must now show that a pure $\alpha$-Nash equilibrium 
        corresponds to a local minimum of \lexmaxksat{4}.
        Contrapositively, we show that a variable flip that improves 
        the weight of satisfied clauses corresponds to a player who will 
        change their strategy. 
        Consider an improving flip of variable $x_i$ (from False to True, without 
        loss of generality). There is an index $j$ such that $C_j$ used to be unsatisfied 
        before the flip and is now satisfied, and the clauses $C_{j+1}, \dots, C_m$ do 
        not change their status. 
        The delay for player $x_i$ before the flip is therefore at least   
        \begin{align}
        \label{cost-before-flip}
            (1 + \alpha)^{j} \ . 
        \end{align}        
        Since the flip does not change the status of
        $C_{j+1}, \dots, C_m$, the unsatisfied among them contain neither $x$ nor $\bar{x}$ 
        and thus do not contribute to the delay for player $x$.
        Thus, their delay after the flip is at most 
        \begin{align}
        \label{cost-after-flip}
         \sum_{l\le j-1}(1 + \alpha)^l = \frac{(1 + \alpha)^{j}-1}{\alpha} \ .    
        \end{align}
        Thus, changing the strategy reduces the delay of player $x_i$ at least by 
        a factor of $\alpha$, so this is not an $\alpha$-Nash equilibrium.
    \end{proof}





    In fact, we can replace the step functions with continuous functions. For this, we add dummy players with only one strategy consisting of occupying a single resource, so that any resource can be occupied by 4 players (effectively this corresponds to making all clauses have 
    length exactly $4$). We define a common function $ f(x)$ and set the delay 
    of the resource $C_j$ under load $x \in \{0,1,2,3,4\}$ 
    to $d_{C_j}(x) = f(mx+j)$. Note that $j < m$, since $m$ is the number of clauses.
    We choose $f$ such that 
    \begin{align}
        f(x+1) \ge (1+\alpha) f(x) 
        \label{f-large-gap}
    \end{align}
    and $f(x) \ge 0$ hold for any $x \in [0, 5m-1]$. The term in (\ref{cost-before-flip}) is then $f(x)$ for some $x < 5m$ and the term in \cref{cost-after-flip} is then $\sum_{k < x} f(k)$ as flipping $x_i$ may lead to all smaller clauses being false (that means 4 players using the clause resource) and in higher-priority clauses that contain $x_i$, the number of unsatisfied literals 
    may increase to some $l \leq 3$.
    This is a gross overestimate.  From 
    (\ref{f-large-gap}) it follows that $f(k-i) \leq (1+\alpha)^{-i}f(x)$ and thus 
    (\ref{cost-after-flip}) is at most 
    \begin{align*}
        \sum_{k < x} f(k) & = \sum_{i=1}^x f(x-i) \leq f(x) \sum_{i=1}^x (1+\alpha)^{-i} 
        \leq f(x) \sum_{i=1}^\infty (1+\alpha)^{-i} = \frac{f(x)}{\alpha} \ .
    \end{align*}
    Therefore this is an $\alpha$-improving move.
    Now setting $f(x) := (1+\alpha)^x$ obviously fulfills~(\ref{f-large-gap}); so does $f(x) := x^d$ 
    with $d \ge 5 m \ln(\alpha+1)$ since
    $(\frac{x+1}{x})^d \geq e^{\frac{d}{x+1}} \geq e^{d/5m} \geq \alpha+1$.

    \section{Permutation Orbit Local Minimum with Abelian groups - Proof of Theorem ~\ref{theorem-abelian-polm}}
    \label{section-abelian-polm}

    We will show a \PLS{}-reduction from \lexmaxksat{4} to Permutation Orbit Local Minimum where all the $\pi_i$ commute. We use a technique by Buchheim and Jünger~\cite{DBLP:journals/disopt/BuchheimJ05}.
    Let the 4-CNF $F = C_1 \wedge \dots \wedge C_m$ be the input instance and 
      $V = \{x_1, \dots, x_n\}$  its variables. Clause $C_i$ has weight $2^{m-i}$.
    We assume for simplicity that 
    each clause in $F$ has exactly four literals; the reduction for the 
    general case is almost identical, but notation would become less readable. We will now describe 
    the instance of \polm{} - set of positions, initial bit string,  the commuting permutations, and, finally,
    the order on the positions.
    
    \begin{enumerate}
    \item \textbf{The set of positions.} 
    For each 4-clause $C$ and each $b \in \{0,1\}^4$ we construct a position $C_b$. 
    There is exactly one assignment $b^{*} \in \{0,1\}^4$ to the four variables 
    in $C$ that violates it; the corresponding position $C_{b^*}$ is called 
    a {\em violating position} and denoted by $C^*$.
    We have a 
    total of $16m$ positions; among those, $m$ are violating. 
    \item \textbf{The initial bit string.}
    Our initial bit string $\vec{s}$ puts a $1$ on $C_{\vec{0}}$ for each clause and a $0$ on 
    all other $C_b$. This corresponds to the assignment that sets 
    all variables to $0$.

    \item \textbf{The permutations.}
    For a variable $x$ and a clause $C$, let $\pi_{x,C}$ be the 
    permutation defined as follows:
    \begin{enumerate}
    \item If $x$ does not show up in $C$, then $\pi_{x,C}$ is the identity.
    \item If $x$ is the $i$-th variable of $C$ (with $1 \leq i \leq 4$), then 
    $\pi_{x,C}$ is the involution that swaps each position $C_b$ with 
    position $C_{b \oplus e_i}$. 
    \end{enumerate}
    So $\pi_{x,C}$ basically switches the value of $x$ {\em in clause $C$}. 
    We define $\pi_{x}$ to be 
    \begin{align*}
        \pi_x := \prod_{C \in F} \pi_{x, C}
    \end{align*}
    so $\pi_x$ switches the value of $x$ in every clause. 
    Our set of generators is 
    \begin{align*}
    G := \{\pi_x \ | \ x \in V\} 
    \end{align*}
    \item \textbf{The ordering of the positions.}
    First come 
    all the violating positions $C^*_1, C^*_{2}, \dots, C^*_m$, using the same 
    priority order as the clauses (with $C_1$ having the highest priority).
    Then come the $15m$ non-violating positions, in some arbitrary order. 
    \end{enumerate}

    Consider a string $\vec{v}$  in the orbit, meaning 
     $\vec{v} = \vec{s} \circ \sigma$ for some $\sigma \in \group{G}$. 
     Then for each clause $C$, exactly one of the sixteen positions 
     $C_b$ has a $1$ in $\vec{v}$; this corresponds to some 
     truth assignment to the variables in $C$. Those assignments agree globally and define 
     an assignment $\alpha_v: V \rightarrow \{0,1\}$. 
    
    \begin{proposition}
    Let $1 \leq j \leq m$ and let $\vec{v} $ be in the orbit. 
    Then $v_j$ is $1$ if and only if the corresponding assignment 
    $\alpha_v$ violates $C_j$.
    \end{proposition}
    
    Thus, if $\alpha_v$ can be improved by flipping some variable $x$,
    then $\vec{v} \circ \pi_x$ decreases the string lexicographically;
    Thus, a locally minimal string in the orbit corresponds to 
    an assignment that cannot be improved by 1-flips i.e., 
    a local optimum for \lexmaxksat{4}.
    This concludes the proof of \cref{theorem-abelian-polm}.
    
\section{Max 4-SAT with lexicographic costs - proof of Theorem~\ref{theorem-lex-4-sat-pls-complete}}\label{section-lex-4-sat}

This is the main technical contribution of the paper. We will proceed in three steps. 
\begin{enumerate}
\item In Section~\ref{subsection-almost-lex}, 
we present a proof by Krentel\cite{DBLP:conf/coco/Krentel89,DBLP:conf/focs/Krentel89} that \localmaxksat{$4$}
is PLS-complete; we will make sure that the clause weights are ``almost lexicographic'', meaning that 
each clause weight is a power of two, and each weight appears for at most two clauses. 
\item In Section~\ref{subsection-lex-double-flip} we modify this construction to achieve fully lexicographic 
weights; this requires a new idea and comes at the cost that we need to allow ``double flips'' for PLS-completeness, 
that is, the neighborhood relation is defined by flipping up to two variables.
\item In Section~\ref{subsection-lex-single-flip} we show how to simulate double flips by a sequence of 
single flips. This will then prove Theorem~\ref{theorem-lex-4-sat-pls-complete}.
\end{enumerate}

\subsection{Krentel's Construction: Reducing \flip{} to \localmaxksat{4} With Almost Lexicographic Clause Weights}\label{subsection-almost-lex}
Let $C'$ be an instance of \flip{}: a Boolean circuit with $n$ inputs and $m$ outputs. We can transform this into a new circuit $C$ with $n$ inputs and $m+n$ outputs that also computes the best neighbor of the current input.
We can assume without loss of generality that $C$ uses only NAND-gates. We create two copies 
of $C$, called TOP and BOTTOM. For every input $1 \leq i \leq n$ of $C$ we create two 
variables $x_i^T$ and $x_i^B$, giving a total of $2n$ input variables. For every gate $j$ of $C$ we create 
variables $g_j^T$ and $g_j^B$. Of those, $2(m+n)$ correspond to the outputs of TOP and BOTTOM. 
It is advantageous to give them synonyms: we denote the $m+n$ output variables of TOP 
$v^T_1,\dots, v^T_m, z^T_1, \dots, z^T_n$ and those of BOTTOM 
$v^B_1,\dots, v^B_m, z^B_1, \dots, z^B_n$. Altogether, if $C'$ has $s$ NAND-gates,  
we get $2s$ gate variables, thus a total of $2(n+s)$ variables.\\

\textbf{Feeding output back to the inputs.} The variables represent a partial evaluation of the two circuits. 
NAND-gates can temporarily be incorrect, for example if $g_j^T$ is a gate variable 
with inputs $y^T_1$ and $y^T_2$ (those being either gate variables or input variables), but 
all three variables are set to $1$, because the correct output of the gate would be $0$. 
Also, we connect the last $n$ output bits (the ``successor output'') of TOP to the input of BOTTOM 
and vice versa. The intuition is that we can 
(1) flip gate variables in TOP to make every gate in TOP correct;
(2) flip input variables in BOTTOM to make them agree with TOP's output; 
(3) flip gate variables in BOTTOM 
to make all gates there correct; 
(4) flip input variables in TOP to make them agree with BOTTOM's output.
By that time most likely some gates of TOP will be incorrect again (by us having meddled with its input variables)
and so we go back to step (1) and repeat. By creating the right clauses and choosing appropriate weights we can make sure 
that each flip increases the weight of the satisfied clauses as long as the value outputs of the circuits increase. \\

\textbf{The active and the inactive circuit.} To make this idea work, we need an additional variable $f$, 
the flip variable, which indicates which circuit we consider to be {\em active}: TOP (if $f=1$) or BOTTOM (if $f=0$). 
The idea is now that gates in the active circuit should be correctly evaluated at all times and 
that $f$ should be flipped only when both circuits are correctly evaluated and the currently inactive circuit 
has better output than the other: if $(v^A_1 \dots v^A_m) >_{\rm lex} (v^I_1 \dots, v^I_m)$ where 
$A \in \{\textnormal{T}, \textnormal{B}\}$ denotes the active circuit and $I$ the inactive circuit.
See Figure~\ref{figure-two-circuits} for an illustration.\\
\begin{figure}
\centering
\includegraphics[width=0.6\textwidth]{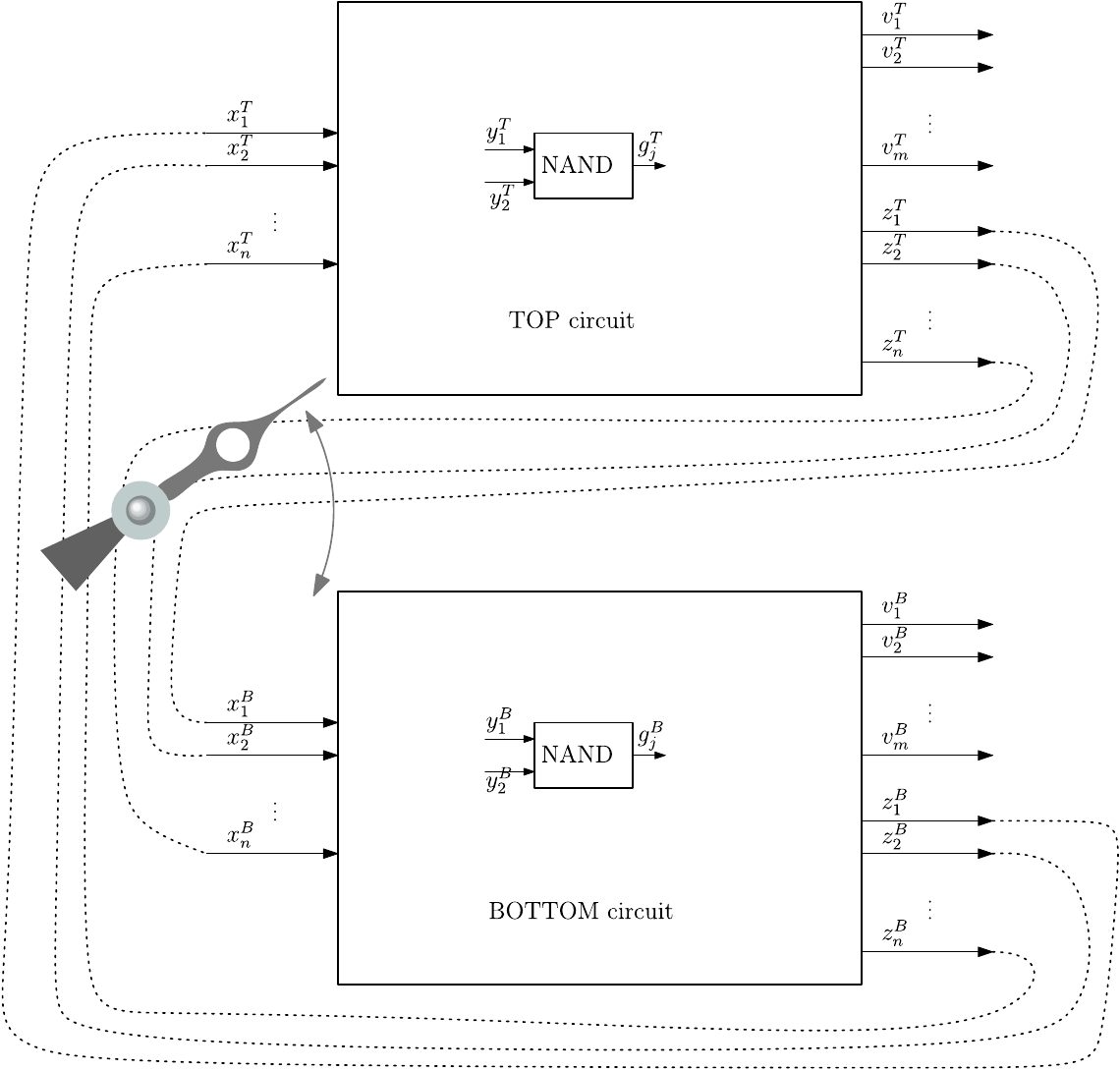}
\caption{The basic layout of Krentel's reduction\cite{DBLP:conf/coco/Krentel89,DBLP:conf/focs/Krentel89}. The big needle represents 
the flip variable $f$. Right now, TOP is active.}
\label{figure-two-circuits}
\end{figure}

\textbf{The clauses.} We will now describe in detail the clauses of our $4$-CNF formula $F$. 
    They come in groups, which we list from most-important to least-important. 

    \begin{enumerate}
    \item \label{point-active-circuit-correct} \textbf{Every gate in the active circuit should be correct.} When a gate $g_j$ has 
    inputs $y_1$ and $y_2$, we write 
    \begin{align}
        f \rightarrow (g_j^T \leftrightarrow \NAND(y_1^T, y_2^T))
        \label{top-active-correct}
    \end{align}
    and express this as a conjunction of four 4-clauses. Similarly, for the case that 
    BOTTOM is active:
    \begin{align}
    \neg f \rightarrow (g_j^B \leftrightarrow \NAND(y_1^B, y_2^B))
    \label{bottom-active-correct}
    \end{align}
    Within this group, the ordering of the clauses follow the topological ordering 
    of the gates (when $g$ feeds into $h$, then $g$ comes before $h$ and the 
    clauses thusly produced come first, too. Within the eight clauses of 
    (\ref{top-active-correct}-\ref{bottom-active-correct}) for a single gate, the order 
    is arbitrary.
    \item \label{point-active-output-high} \textbf{The value output of the active circuit should be large.} 
    For each value output position  $1 \leq i \leq m$ we create two 2-clauses:
    $(f \rightarrow v^T_i)$ and $(\neg f \rightarrow v^B_i)$. 
    The clauses are ordered from $i=1$ (most important) to $i=m$ (least important), 
    and for the same $i$, the two clauses have equal weight (here we lose perfect lexicographicality).
    
    \item \label{point-feedback-clauses} \textbf{Successor output of active circuit should be input of inactive circuit.}
    We write this as 
    \begin{align}
    \label{eqn-active-output-is-inactive-input}
    f \rightarrow (z_j^T = x_j^B)
    \textnormal{ and }
    \neg f \rightarrow (z_j^B = x_j^T) \  . 
    \end{align}
    The ordering within this group is arbitrary.
    \item \label{point-inactive-correct} \textbf{All gates should be correctly evaluated,} whether in the active or inactive circuit.
    \begin{align}
    g_j^T \leftrightarrow \NAND(y_1^T, y_2^T) \text{\ and\ }
    g_j^B \leftrightarrow \NAND(y_1^B, y_2^B) 
    \end{align}
    \end{enumerate}

    We call clauses in Group~\ref{point-active-circuit-correct} {\em hard clauses} 
    because they would always stay satisfied on an improving path, and if 
    violated can be easily corrected.
    It is straightforward to check that in a local optimum of this \localmaxksat{4}-instance, 
    all gates are evaluated correctly and the output of the active circuit agrees with the input 
    of the inactive. Thus, it holds that $(\mathbf{v}^A, \mathbf{z}^A) = C(\mathbf{x}^A)$ and
    $(\mathbf{v}^I, \mathbf{z}^I) = C(\mathbf{x}^I) = C(\mathbf{z}^A)$.
    If $\mathbf{v}^I >_{lex} \mathbf{v}^A$, then flipping the value of $f$ would 
    (i) keep clauses in Group~\ref{point-active-circuit-correct} satisfied;
    (ii) improve the weight of satisfied clauses in Group~\ref{point-active-output-high}, 
    and (iii) possibly violate clauses in  Group~\ref{point-feedback-clauses}. 
    By the chosen weights, this would be an improving flip. Thus, if we have a local optimum of 
    this \localmaxksat{4}-instance, then $\mathbf{v}^I \leq_{lex} \mathbf{v}^A$ and 
    $\mathbf{x}^A$ is indeed a solution of the original \flip{}
    instance.

\subsection{Making Weights Fully Lexicographic By Allowing Double Flips}
\label{subsection-lex-double-flip}

The trouble lies in Group~\ref{point-active-output-high}. If we naively decree that 
$(f \rightarrow v_i^T)$ be more valuable than (actually, twice as) $(\neg f \rightarrow v_i^B)$, 
then the reduction would become incorrect: if $\mathbf{v}^T = (01\dots)$ and $\mathbf{v}^B = (00\dots)$
and BOTTOM is active, then flipping $f$ from $0$ to $1$ would improve the output but violate $(f \rightarrow v_i^T)$
instead of $(\neg f \rightarrow v_i^B)$, thus reducing the weight of satisfied clauses. We are stuck in a local 
optimum that does not correspond to a local optimum of \flip{}. \\

To solve this problem, we create a {\em global output register} $\mathbf{o}$ consisting of $m$ variables $o_1, \dots, o_m$. 
The idea is now that (1) $\mathbf{o}$ should be large, but (2) not larger than the active output. When switching 
active and inactive circuits, we can flip $f$ {\em and simultaneously} the improving position of $\mathbf{o}$,
improving (1) but temporarily violating (2). To make this work we need to choose the order of the 
clauses in Group~\ref{point-active-output-high} carefully.

\begin{enumerate}
\setcounter{enumi}{1}
\item \textbf{The global output register should be large, but not larger than the active output.} 
For each $1 \leq i \leq m$, we replace the two old clauses $(f \rightarrow v^T_i)$ and $(\neg f \rightarrow v^B_i)$ by
\begin{align}
       \textnormal{the $\nth{i}$ TOP control clause } C_i^T :=\ & (f \rightarrow (o_i \leq v_i^T)) \label{control-top} \\
        \textnormal{the $\nth{i}$  BOTTOM control clause } C_i^B :=\ & (\neg f \rightarrow (o_i \leq v_i^B)) \label{control-bottom} \\
        \textnormal{the $\nth{i}$  output clause } C_i^O :=\  & (o_i) \label{payoff}
    \end{align}
\end{enumerate}
The $m$ blocks are ordered from $1$ (most important) to $m$ (least important) and within 
each block, the three clauses\footnote{$C_i^T$ and $C_i^B$ are indeed clauses: $(\neg f \vee \neg o_i \vee v_i^T)$
and $(f \vee \neg o_i \vee v_i^B)$.}
are ordered as shown above. \\

Now if we have a local optimum of this \lexmaxksat{4} instance, then by the same reasoning as in 
Section~\ref{subsection-almost-lex}, all circuits are correctly evaluated and BOTTOM's input is TOP's output and 
thus $(\mathbf{v}^A, \mathbf{z}^A) = C(\mathbf{x}^A)$ and
$(\mathbf{v}^I, \mathbf{z}^I) = C(\mathbf{x}^I) = C(\mathbf{z}^A)$.
Additionally, it holds that $\mathbf{o} = \mathbf{v}^A$. 
If $\mathbf{v}^I \leq_{lex} \mathbf{v}^A$ we have found a solution to the original \flip{} instance.
Otherwise, $\mathbf{v}^I >_{lex} \mathbf{v}^A = \mathbf{o}$. Let $i$ be the first position where they differ, so $v^I_i = 1$ and $v^A_i = 0$ (and thus $o_i =0)$. 
Now flipping the value $f$ and simultaneously 
setting $o^i$ to $1$ keeps Group~\ref{point-active-circuit-correct} satisfied; 
the clauses in blocks $1, \dots \leq i-1$ of Group~\ref{point-active-output-high} do not change; 
in block $i$ we gain $(o_i)$ and now satisfy all three clauses;
we might wreak havoc in blocks $i+1, \dots, m$ and Group~\ref{point-feedback-clauses}, 
but gaining $(o_i)$ outweighs all this. Thus, 
if we are not at a local optimum for \flip{}, we aren't at one for \lexmaxksat{4}, either, 
and the reduction is correct.

\subsection{Replacing double flips with single flips}\label{subsection-lex-single-flip}

In the previous section, we need to flip two variables to simulate a flip in the circuit instance. Flipping only $o_i$ from $0$ to $1$ will violate the $\nth{i}$ control clause 
of the active circuit, and flipping  $f$ alone does not increase and likely 
reduces the payoff. We introduce new variables that we call {\em shadow variables}, denoted by $s_i^B$ and $s_i^T$ for every output bit. Their purpose is to temporarily fulfill a control clause so that we can flip $o_i$. Then we can flip $f$ and continue as above.

We need to introduce two new groups of hard clauses, 0 and 1b, the latter located between group 1 and 2; modify group 2; and add groups 5 and 6. Groups 1, 3, and 4 are 
as above.

\begin{itemize}
    \item[\textbf{0.}] \textbf{Both inputs are equal or neighbors of the flip instance.}
    These hard clauses prevent assignments where the inputs of both circuits have a Hamming distance greater than one. They are not strictly needed for \PLS{}-completeness 
    but ensure {\em tightness} of the reduction (see \cref{section-tightness}).
    \begin{equation}
    \label{eq:Hamming-distance-at-most-one}
        \forall i\ne j: \ x^T_i \ne x_i^B  \to x_j^T = x_j^B
    \end{equation}
    The ordering among Group 0 clauses is arbitrary.

    \item[\textbf{1b.}]
        \textbf{Proper use of shadow variables.}
        These hard clauses have the purpose to disallow using shadow variables if there is no flip to be done. All clauses in this group can be satisfied by setting a shadow variable to false.
    
        Shadow variables 
        should be turned on only when the corresponding circuit is correctly 
        evaluated and, in case it is the inactive circuit, it is reading 
        its correct input, namely the best-neighbor output of the active circuit.
        Thus, for every output position $1 \leq i \leq m$ and every 
        input position $1 \leq j \leq n$ we produce
        \begin{align}
        \label{eq:shadow-variable-correct-input}
        (s^B_i \wedge f) \rightarrow (z_j^T = x_j^B) \text{\ and \ }
        (s^T_i \wedge \neg f) \rightarrow (z_j^B = x_j^T) 
        \end{align}
        saying that a ``shadowed inactive'' circuit must read the 
        correct input. Second, we require a shadowed circuit to be 
        correct in all gates. Thus, if a gate $g_j$ has inputs 
        $y_1$ and $y_2$, we produce
        \begin{align}
            \label{eq:shadow-variable-correct-circuit}
            &s^T_i \rightarrow (g_j^T \leftrightarrow \NAND(y_1^T, y_2^T)) \text{\ and \ }
            s^B_i \rightarrow (g_j^B \leftrightarrow \NAND(y_1^B, y_2^B))
        \end{align}
        Lastly, shadow variables  must not be used in vain---meaning 
        they should be used only when they indicate a position where 
        the inactive output is better than the active one. 
        Why is this necessary? When $i$ is the most significant bit in 
        which the active output is 0 but the inactive is 1 (so $i$ is where 
        we can improve things), we can be sure that all shadow variables 
        of the positions $1, ..., i-1$ are off, which will be important later.
        \begin{align}
            \label{eq-shadow-variable-necessary}
            s^B_i \to v_i^B \text{\ and \ }
            s^B_i \to \neg v_i^T \\
            s^T_i \to v_i^T \text{\ and \ }
            s^T_i \to \neg v_i^B 
        \end{align}
        The order among the clauses in this group is arbitrary. 
        All clauses so far (correctness clauses and proper use clauses)     
        stay fulfilled along the canonical path and are never violated when only doing improving steps. If violated, we can always satisfy them by correcting the input of the inactive circuit (group 0), correcting the output of gates (group 1), 
        or by deactivating shadow variables (group 1b).

    \item[\textbf{2.}] \textbf{The global output register should be large, but not larger than the active output.}\\
    The first three clauses are almost exactly the same as in the previous section, except that control clauses can also be fulfilled by the shadow variables of the other circuit.
    The last clauses are used to switch the active circuit after the global output was improved via a shadow variable.
    \begin{align}
        C_i^T :=\ &f \rightarrow (o_i \le v_i^T) \lor s^B_i \label{eq:control-clause-top}\\
        C_i^B :=\ &\neg f \rightarrow (o_i \le v_i^B) \lor s^T_i \label{eq:control-clause-bottom}\\
        C_i^O :=\ &o_i \label{eq:output-clause}\\
        &(o_i \land s_i^B) \rightarrow \neg f
        \label{eq:shadowActSwitch}\\
        &(o_i \land s_i^T) \rightarrow f \label{eq:shadowActSwitch2} 
    \end{align}
    
    \item[\textbf{5.}] \textbf{Use shadow variables at improvable positions.}
    If the inactive circuit is better at a position $i$ than the active circuit, these clauses will be false unless the shadow variable is turned on. Shadow variables of the active circuit have, however, no incentive to be switched on from here. The order is arbitrary in this group.
    \begin{align}\label{eq-active-shadow-variable}
        (f \land v_i^B \land \neg v_i^T) \rightarrow s_i^B \text{\ and \ } (\neg f \land v_i^T \land \neg v_i^B ) \rightarrow s_i^T
    \end{align}
    \item[\textbf{6.}] \textbf{Turn off shadow variables}\\
    The final set of clauses allows deactivating any shadow variable. The order is arbitrary in this group.
    \begin{align}       
        \neg s_i^B \text{\ and\ }  \neg s_i^T\label{eq:negshadow}
    \end{align}
\end{itemize}

Previously, we flipped both $f$ and a global output $o_i$ in order to gain an improvement. The new idea is to first set the shadow 
variable $s_i^B$ to $1$ (without loss of generality BOTTOM is inactive), for which 
clause (\ref{eq-active-shadow-variable}) gives us an incentive. Now the control clause 
$f \rightarrow (o_i \le v_i^T) \lor s^B_i$ (\ref{eq:control-clause-top}) is additionally satisfied 
by $s^B_i$, which enables us to flip the global output $o_i$ to $1$; finally this leads to flipping $f$. As a last step we can turn all shadow variables off again in the now active circuit BOTTOM using group 6. \PLS-completeness now follows from the 
following lemma, which we prove in  \cref{section-4sat-appendix}:

\begin{lemma}
  Let $\alpha$ be a locally optimal assignment to the 4-CNF just constructed and consider 
  $\mathbf{x}^A$, the input variables of the active circuit. Then $\alpha(\mathbf{x}^A)$ is a local 
  optimum of the \flip{} instance $C$. 
\end{lemma}

\section{Conclusion and open problems}

We added two members to the club of \PLS{}-complete problems with lexicographic costs: \lexmaxksat{4} and Abelian \polm{}.
We believe that as a point for further reductions, \lexmaxksat{4} is better than 
\flip{} due to its simple structure while preserving the lexicographical costs. We 
demonstrated this by showing that finding approximate Nash equilibria is \PLS{}-complete,
and the simple structure of lexicographic costs allowed us to use nice continuous delay functions. Additionally, all reductions in this work are tight (see \cref{section-tightness}).

Multiple open problems remain: While \lexmaxksat{$k$} is \PLS{}-complete
for $k=4$ and in \P{} for $k=2$, the case $k=3$ is open. 
Our reduction cannot easily be changed to produce $3$-clauses:
we already need 3-clauses to 
model the correct evaluation of a circuit (2-clauses don't suffice because 
2-SAT is in \NL{} and circuit evaluation is \P{}-complete); and then we need 
an additional variable $f$ to temporarily de-activate one copy of the circuit. 
It is also open whether \localmaxksat{4} stays \PLS{}-complete if every variable
appears a bounded number of times.

Secondly, in our reduction and the reduction in \cite{SchederTantow2025} for \polm{} the number of permutations grows with the input size. What happens if we only have a fixed number of permutations $k$? For $k=1$ this is in P~\cite{DBLP:conf/sat/KolodziejczykT24,SchederTantow2025} but even for $k=2$ it is open.

Thirdly, what happens if the delay functions in the congestion game are polynomials 
of bounded degree $d$ (and perhaps weighted)? While upper bounds are known, there are no lower bounds known for those $\alpha$ where the problem is in \PLS\cite{DBLP:journals/teco/CaragiannisFGS15, DBLP:journals/mor/ChristodoulouGGPW23, DBLP:journals/mor/GiannakopoulosN22}. In \cite{DBLP:journals/mor/ChristodoulouGGPW23} it is shown that for $\alpha \in \mathcal{O}(\frac{d}{\ln d})$ deciding the existence of an approximate Nash-equilibrium is \NP{}-hard if the game is weighted.  \lexmaxksat{4} might be a helpful intermediate problem to reach this target. 

Lastly, another interesting \PLS{} problem with lexicographic costs might be the traveling salesman problem as here it is also not feasible to check if there exists a solution using a set of edges. Known reductions\cite{DBLP:conf/icalp/Heimann0H24, DBLP:journals/siamcomp/SchafferY91} start however from Max-Cut and are therefore not adaptable for lexicographic costs.

\bibliography{main}

\appendix

\section{Remaining proofs for Max-4-SAT}\label{section-4sat-appendix}
    In this section, we give the remaining proofs for \cref{subsection-lex-single-flip}.
    The reduction given above transforms a circuit $C$ (an instance of \flip{}) 
    into a 4-CNF-formula $F = F(C)$ (an instance of \localmaxksat{4}).
    One step was to augment the circuit $C$ to not only output the $m$-bit 
    vector representing the {\em value} of its input $\mathbf{x}$, but also 
    the best neighbor of $\mathbf{x}$. We denote by $C_V(\mathbf{x}) \in \{0,1\}^m$ the 
    value and by $C_N(\mathbf{x})$ the best neighbor. 
    Let $\alpha$ be a truth assignment to the variables of $F$ and $f$ the 
    {\em flip variable}. Recall that we call the TOP circuit {\em active} if 
    $\alpha(f) = 1$; we call BOTTOM active if $\alpha(f) = 0$.
    We call the other circuit {\em inactive}.
    We define a function $g$ that maps assignments $\alpha$ of $F$
    to circuits input for $C$ as follows:

    \begin{align*}
        g(\alpha) = \begin{cases}
            & \alpha(x^B) \text{\ ,\ if $\alpha(f) = 0$ }\\
            & \alpha(x^T) \text{\ ,\ if $\alpha(f) = 1$ }\\
        \end{cases}
    \end{align*}

    \begin{theorem}
    \label{theorem-4-sat-appendix}
    Let $\alpha$ be a locally optimal assignment to the 
    variables of $F$. Then $g(\alpha)$ is a solution to the flip instance $C$.
    \end{theorem}
    \begin{proof}
        We will assume without loss of generality 
        that $\alpha(f)=1$, i.e., the TOP circuit is active.

    \begin{lemma}\label{lem-4sat-shadowvar-active-false}
        $\alpha$ maps all shadow variables of the active circuit TOP to $0$.
    \end{lemma}
    \begin{proof}
        Suppose, for the sake of contradiction, that  
        $\alpha(s_i^T)=1$ for some $1 \leq i \leq m$. What happens 
        when we change $\alpha(s_i^T)$ to $0$?
        Clause groups 0, 1, 3 and 4 do not contain shadow variables.
        Clauses of groups 1b and 6 contain only the negative literal 
        $\neg s^T_i$, so the change can only improve. 
        In group 2, 
        $s_i^T$ occurs in 
        $C_i^B :=\ \neg f \rightarrow (o_i \le v_i^B) \lor s^T_i$
        (\ref{eq:control-clause-bottom})
        and 
        $(o_i \land s_i^T) \rightarrow f$ (\ref{eq:shadowActSwitch2}), and in group 5
        only in $(\neg f \land v_i^T \land \neg v_i^B ) \rightarrow s_i^T$ 
        (\ref{eq-active-shadow-variable}). All three clauses are and stay satisfied by 
        $\alpha(f)=1$.
         Group 6 contains $(\neg s_i^T)$ (\ref{eq:negshadow}), so the satisfied weight 
        improves, meaning that $\alpha$ was no local optimum.        
    \end{proof}

    \begin{lemma}\label{lem-4sat-active-circuit-correct}
        In $\alpha$, all gate variables in the active circuit are correct, i.e. their output matches what an actual {\rm NAND} would do for the current values of the inputs.
    \end{lemma}
    \begin{proof}
        Suppose not, so there is a gate variable $g_j^T \ne {\rm NAND}(y^T_{1}, y^T_{2})$,
        where $y^T_{1}, y^T_{2}$ are the variables feeding into gate $g_j$ (either 
        earlier gates or input variables). Changing the value of $g_j^T$ leaves 
        clauses in group 0 unaffected; in group 1 we satisfy an additional 
        clause of 
        $f \rightarrow (g_j^T \leftrightarrow \NAND(y_1^T, y_2^T))$ (\ref{top-active-correct}) 
        belonging to gate $j$ in TOP while not violating any gate coming before $j$; we might violate lower-priority clauses.
        The total satisfied weight improves and $\alpha$ was not a local optimum.        
    \end{proof}

    \begin{lemma}\label{lem-4sat-if-shadow-then-correct}
        Suppose $\alpha(s^B_i)=1$ for some shadow variable $s_i^B$. Then all 
        gate variables in {\rm BOTTOM} are correct.
    \end{lemma}
    \begin{proof}
        Otherwise, if some gate variable $g_j^B$ were incorrect, 
        we could simply change $\alpha(s^B_i)$ to $0$; this 
        would not affect groups 0 and 1; in group 1b the currently 
        violated clause $s^B_i \rightarrow (g_j^B \leftrightarrow \NAND(y_1^B, y_2^B)$ 
        (\ref{eq:shadow-variable-correct-circuit}) 
        would become satisfied, and no 
        additional clause in group 1b would become violated, so this 
        would improve the satisfied weight, and $\alpha$ would not be a local 
        optimum.
    \end{proof}

    \begin{lemma}\label{lem-4sat-circuit-correct}
        In $\alpha$, all gate variables in the inactive circuit are correct.
    \end{lemma}
    
    \begin{proof}   
        If some $\alpha(s_i^B)=1$ this follows from Lemma~\ref{lem-4sat-if-shadow-then-correct}. So we can assume (*) that $\alpha(s_i^B)=0$ for all 
        shadow variables. For the sake of contradiction, suppose some 
        gate variable $g_j^B$ is incorrect:
        $g^B_j \ne {\rm NAND} (y^B_{1}, y^B_{2})$. What happens when we flip 
        $g^B_j$?
        Clauses in group 0 are unaffected; clauses in group 1 concerning the bottom circuit 
        are of the form $\neg f \rightarrow (g_{j'}^B \leftrightarrow \NAND(y_3^B, y_4^B))$
        (\ref{bottom-active-correct}) and thus are and stay satisfied by $\alpha(f)=1$.
        Group 1b stays satisfied because all 
        shadow variables are $0$ (in TOP by Lemma~\ref{lem-4sat-shadowvar-active-false}; 
        in BOTTOM by assumption (*)). 
  If $g^B_j$ is a value output variable $v_i^B$ then group 2 contains
  $\neg f \rightarrow (o_i \le v_i^B) \lor s^T_i$ (\ref{control-bottom});
  if $g^B_j$ is a best-neighbor output variable $z_{j'}^B$ then group 3 contains
  clause $\neg f \rightarrow (z_{j'}^B = x_{j'}^T)$
  (\ref{eqn-active-output-is-inactive-input}); both are and stay 
  satisfied by $\alpha(f)=1$. 
   In group 4, changing the value of $g_j^B$ 
   will satisfy all of $g_j^B \leftrightarrow \NAND(y_1^B, y_2^B)$, of which 
   at least one was previously violated, and violate only lower-priority clauses, 
   due to the clauses being ranked by the  topological order on the gates. 
   Thus, we improve the total 
   weight, meaning that $\alpha$ was not a local optimum.
    \end{proof}

    Let us summarize: all shadow variables in TOP are $0$, and all gates 
    are correctly evaluated.
    Formally: $\alpha(\mathbf{z}^T) = C_N(\alpha(\mathbf{x}^T))$;
    $\alpha(\mathbf{v}^T) = C_V(\alpha(\mathbf{x}^T))$;
    $\alpha(\mathbf{z}^B) = C_N(\alpha(\mathbf{x}^B))$;
    $\alpha(\mathbf{v}^B) = C_V(\alpha(\mathbf{x}^B))$.
    Next, we show that BOTTOM's input is equal to TOP's output.

    \begin{lemma}\label{lem-4sat-if-shadow-then-input-correct}
        If $\alpha(s^B_i)=1$ then $\alpha(\mathbf{x}^B) = \alpha(\mathbf{z}^T)$.
    \end{lemma}
    \begin{proof}
        If not, then $\alpha(x^B_j) \ne \alpha(z^T_j)$ for some $1 \leq j \leq n$.
        Then clause $(s^B_i \wedge f) \rightarrow (z_j^T = x_j^B)$ (\ref{eq:shadow-variable-correct-input}) 
        from group 1b
        is violated. If we change $s^B_i$ to $0$ we satisfy it and do not violate 
        any other clause of group 1b since $s^B_i$ only occurs negatively here; 
        group 0 does not contain $s^B_i$ at all. 
    \end{proof}

    \begin{lemma}\label{lem-4sat-input-correct}
        $\alpha(\mathbf{x}^B) = \alpha(\mathbf{z}^T)$.
    \end{lemma}
    \begin{proof}
        If some $s^B_i$ is set to one, this follows from 
        \cref{lem-4sat-if-shadow-then-input-correct}. So we can assume 
        that all shadow variables are $0$. For the sake of contradiction, 
        suppose there is some $1 \leq j \leq n$ with 
        $\alpha(z_j^T) \ne \alpha(x_j^B)$. We distinguish three cases.\\

        \textbf{Case 0: } $\alpha(\mathbf{x}^B)$ and $\alpha(\mathbf{x}^T)$ are equal.
        Let us change the value of $x_j^B$. Group 0 stays satisfied because now exactly one 
        input bit differs; group 1 stays satisfied because $\alpha(f)=1$; 
        group 1b stays satisfied because all shadow variables are $0$.
        Group 2 does not contain input variables $x_j^B$.\footnote{We assume that 
        every input of the circuit passes through at least one gate, so no input 
        variable doubles as an output variable.}
        Group 3 contains
        $f \rightarrow (z_j^T = x_j^B)$ (\ref{eqn-active-output-is-inactive-input}), 
        which are actually two clauses; $\alpha$ violates exactly one of them, 
        and changing $x^B_j$ to $0$ satisfies both. Other clauses in group 3
        do not contain $x_j^B$. This is an improving step.
    
        \textbf{Case 1: } There is exactly one $1 \leq i \leq n$ with 
        $\alpha(x_{i}^T) \ne \alpha(x_{i}^B)$. 
        As argued above, $\alpha(\mathbf{z}^T)$ really is the correct 
        output $C_N(\alpha(\mathbf{x}^T))$. By design, it's the best 
        neighbor of $\alpha(\mathbf{x}^T)$. So both 
        $\alpha(\mathbf{x}^B)$ and $\alpha(\mathbf{z}^T)$ have Hamming 
        distance $1$ to $\alpha(\mathbf{x}^T)$ (the former by assumption of Case 1;
        the latter by design). So there are actually exactly two positions 
        where $\alpha(\mathbf{z}^T)$ and $\alpha(\mathbf{x}^B)$ differ: $i$ and $j$.
        Changing $x_i^B$ instead of $x_j^B$ makes $\mathbf{x}^T$ and $\mathbf{x}^B$ equal, which keeps
        group 0 and 1 satisfied and does not 
        affect group 2; in group 3 it additionally satisfies 
        clauses $f \rightarrow (z_i^T = x_i^B)$ (\ref{eqn-active-output-is-inactive-input})
        (strictly speaking it satisfies both those clauses whereas one of them used to be violated).
        It might affect group 1b, but only in a beneficial way: 
        making $z_j^B$ and $x_j^T$ agree can only improve things here. \\
    
        \textbf{Case 2: } $\alpha(\mathbf{x}^B)$ and $\alpha(\mathbf{x}^T)$ differ         
        in more than one coordinate, say $i$ and $j$. Then a clause coming from 
        $x^T_i \ne x_i^B  \to x_j^T = x_j^B$ (\ref{eq:Hamming-distance-at-most-one}) in group 0 
        is violated and flipping the 
        violating input position $i$ surely improves the weight.
    \end{proof}
    
    We now know that under $\alpha$ the following hold:
    $\mathbf{z}^T$ equals $\mathbf{x}^B$, 
    and it is the best neighbor of $\mathbf{x}^T$. 
    $\mathbf{v}^T$ is the value of $\mathbf{x}^T$ 
    and $\mathbf{v}^B$ is the value of $\mathbf{x}^B$.
    If $\alpha(\mathbf{v}^T) \geq \alpha(\mathbf{v}^B)$, then
    $\alpha(\mathbf{x}^T)$ is at least as good as its best neighbor, is 
    a solution to \flip{}, and we are done. Otherwise, there exists
    $1 \leq i \leq n$ with $\alpha(v^T_i)=0$ and $\alpha(v^B_i)=1$ 
    and $\alpha(v^T_{j}) = \alpha(v^B_{j})$ for all $1 \leq j < i$. 
    We will derive a contradiction.
    
    \begin{lemma}
    \label{lemma-o-top-bottom-agree-to-the-left}
        Under $\alpha$ we have $o_j = v_j^T = v_j^B$ for all $1 \leq j < i$.
    \end{lemma}
    \begin{proof}
        We already know that $\alpha(s_j^T)=0$ by Lemma~\ref{lem-4sat-shadowvar-active-false}. We claim that $\alpha(s_j^B)=0$, too: If $\alpha(s_j^B)$ were $1$,
        then changing it to $0$ leaves groups 0 and 1 unaffected; in group 1b 
        clauses (\ref{eq:shadow-variable-correct-input}) and
        (\ref{eq:shadow-variable-correct-circuit}) are and stay satisfied;
        among (\ref{eq-shadow-variable-necessary}) exactly one is violated 
        under $\alpha(s_j^B)=1$ and both are satisfied 
        under $\alpha(s_j^B)=0$ (which one depends on whether $v_j^B, v_j^T$ are both 
        $1$ or both $0$). So changing $\alpha(s_j^B)$ to $0$ would improve the 
        satisfied weight. 
        Thus, both $s_j^B$ and $s_j^T$ are $0$ under $\alpha$. 
        
        Next we claim that $\alpha(o_j) = \alpha(v_j^T)$. Suppose not. 
        If $\alpha(o_j)=0$ and $\alpha(v_j^T)=1$ then changing $o_j$ to $1$ 
        wins us clause $(o_j)$ (\ref{eq:output-clause}); if 
        $\alpha(o_j)=1$ and $\alpha(v_j^T)=0$ then changing $o_j$ to $1$ 
        wins us $(f \rightarrow (o_j \leq v_j^T) \vee s_j^B)$ (\ref{eq:control-clause-top}).
        All other clauses in block $j$ of group 2 remain as they were; clauses 
        in other blocks or other groups do not even contain the variable $o_j$. 
    \end{proof}

    Let us summarize: all circuits are correctly evaluated and BOTTOM reads as 
    input the output of TOP. We have $\alpha(v_i^T)=0$ 
    and $\alpha(v_i^B)=1$; all TOP shadow variables $s^T_j$ are $0$; 
    for BOTTOM, the shadow variables $s^B_j$ for $1 \leq j \leq i-1$ are all $0$.    
        We consider four cases, depending on the values of $o_i$ and $s_i^B$:\\

        \textbf{Case 10: } $\alpha(o_i)=1$ and $\alpha(s_i^B)=0$. Then 
        the control clause $(f \rightarrow (o_i \le v_i^T) \lor s^B_i)$
        (\ref{eq:control-clause-top}) is violated; changing $o_i$ to $0$ 
        satisfies it and only loses the lesser clause $(o_i)$ (\ref{eq:output-clause}).\\

        \textbf{Case 00: } $\alpha(o_i)=0$ and $\alpha(s_i^B)=0$. 
        Then $(f \wedge v_i^B \wedge \neg v_i^T) \rightarrow s_i^B$
        (\ref{eq-active-shadow-variable}) is false and we can change $s_i^B$ to $1$ 
        and satisfy it. This does not change any other clause in group 5, 4, 3 and 1 since $s_i^B$ does not occur there. In group 1b all clauses stay fulfilled since by \cref{lem-4sat-circuit-correct} all gate variables are correct and the input of the inactive circuit BOTTOM matches the neighbor output of the active circuit TOP by \cref{lem-4sat-input-correct}. Also since $\alpha(v_i^B)=1$ and  $\alpha(v_i^T)=0$
        holds, the final type of clause in group 1b is fulfilled: 
        $s^B_i \to v_i^B$ and $s^B_i \to \neg v_i^T$~(\ref{eq-shadow-variable-necessary}).
        
        In group 2, $s_i^B$ occurs in 
        $f \rightarrow (o_i \le v_i^T) \lor s^B_i$ (\ref{eq:control-clause-top}) 
        and $(o_i \land s_i^B) \rightarrow \neg f$ (\ref{eq:shadowActSwitch}). 
        The former is satisfied with $\alpha(s^B_i)=1$, the latter 
        is and stays satisfied by $\alpha(o_i)=0$.\\

        \textbf{Case 01: } $\alpha(o_i)=0$ and $\alpha(s_i^B)=1$. Consider setting
        $o_i$ to $1$. This affects only block $i$ of group 2. 
        The control clause
        $f \rightarrow (o_i \le v_i^T) \lor s^B_i$ (\ref{eq:control-clause-top}) 
        is and stays satisfied by $\alpha(s^B_i)=1$; 
        the control clause 
        $\neg f \rightarrow (o_i \le v_i^B) \lor s^T_i$ (\ref{eq:control-clause-bottom})
        is and stays satisfied by $\alpha(f)=1$.  The third clause in block i, 
        $(o_i)$~(\ref{eq:output-clause}), becomes newly satisfied and thus 
        the overall satisfied weight improves.\\

        \textbf{Case 11: } $\alpha(o_i)=1$ and $\alpha(s_i^B)=1$. Then changing 
        $f$ to $0$ wins us the flip clause 
        $(o_i \wedge s_i^B) \rightarrow \neg f$ 
        of group 2~(\ref{eq:shadowActSwitch}). Clauses in earlier blocks $j < i$ of group 2 
        are not affected: control clauses like
        $\neg f \rightarrow (o_j \le v_j^B) \lor s^T_i$~(\ref{eq:control-clause-bottom}) stay satisfied since  
        $o_j$ and $v_j^B$ agree due to Lemma~\ref{lemma-o-top-bottom-agree-to-the-left};
        flip clauses $(o_j \land s_j^T) \rightarrow f$ (\ref{eq:shadowActSwitch2}) 
        stay satisfied because
        all shadow variables of TOP are $0$ by Lemma~\ref{lem-4sat-shadowvar-active-false}.
        Also, we know by \cref{lem-4sat-circuit-correct} that all gate clauses are correct, so  all clauses in group 1 stay correct, too.
        Finally $f$ does not occur in any clause of group 0 or 1b, so they stay satisfied.\\

    Note that this flip process would also start when $v_i^T > v_i^B$ 
    but  $v_k^T < v_k^B)$ for some $k > i$.  It would flip $s_k^B$ to $1$ (Case 00) and then $o_k$ (Case 01) but 
    would fail to flip $f$ (Case 11) since that would violate the control clause
    $\neg f \rightarrow (o_i \le v_i^B) \lor s^T_i$ (\ref{eq:control-clause-bottom}). 
    
    We conclude that $\alpha(\mathbf{x}^T)$ is indeed a solution to the 
    \flip{} instance $C$. This concludes the proof of Theorem~\ref{theorem-4-sat-appendix}.
    \end{proof}

\section{The Max-3-SAT/2-Flip problem}\label{section-lex3sat}
    There is currently a gap between the \PLS{}-completeness of \lexmaxksat{4} and the polynomial algorithm for \lexmaxksat{2}. We can only show hardness for \lexmaxksat{3} if we allow double flips. 
    
    \begin{theorem}\label{theorem-lex-3-sat-pls-complete}
        \lexmaxksat{3}/2-FLIP is \PLS{}-complete.
    \end{theorem}

    The reduction is almost identical to the one in \cref{subsection-lex-double-flip}, with one crucial modification: for each gate $g_j$ we 
    introduce not one but {\em three} variables $g_j, g_{j,1}, g_{j,2}$, and then 
    two copies of each corresponding to TOP and BOTTOM: 
    $g_j^T, g^T_{j,1}, g^T_{j,2}, g^B_j, g^B_{j,1}, g^B_{j,2}$. The idea is that $g_{j,1}$ and $g_{j,2}$ represent
    the two inputs of the gate and $g_j$ the output. We 
    introduce a ``Group 0'' of $3$-clauses with the highest priority that 
    state gates should be correct at all times:
    \begin{align*}
    g_j^T \leftrightarrow \NAND(g_{j,1}, g_{j,2}) \ . 
    \tag{Group 0}
    \end{align*}
    We modify the clauses in Group 1; we replace the 4-clauses 
    constructed in (\ref{top-active-correct}) and (\ref{bottom-active-correct}), 
    namely 
    \begin{align*}
    f \rightarrow (g_j^T \leftrightarrow \NAND(y_1^T, y_2^T)) \\        
    \neg f \rightarrow (g_j^B \leftrightarrow \NAND(y_1^B, y_2^B))
    \tag{The Group 1 in the previous reduction}
    \end{align*}
    by clauses requiring that the input of gate $g_j$ must agree with 
    the output of the gate feeding into it; in other words, if 
    the output of gate $g_{i}$ feeds into the first output of gate $g_j$, we 
    add the following $3$-clauses to Group 1:
    \begin{align*}
    f \rightarrow (g_i^T \leftrightarrow g_{j,1}^T) \\ 
    \neg f \rightarrow (g_i^B \leftrightarrow g_{j,1}^B) \tag{The new Group 1}
    \label{wire-clauses}
    \end{align*}
    and similarly for the second input $g_{j,2}$; we call those 
    clauses {\em wire clauses} because they check whether 
    the value in the circuit are consistent along the wires.
    Also, if the first input of the gate $g_j$ is an input $x_i$ to the circuit, 
    we produce the clause
    \begin{align*}
    f \rightarrow (x_i^T \leftrightarrow g_{j,1}^T) \\ 
    \neg f \rightarrow (x_i^B \leftrightarrow g_{j,1}^B) \tag{The new Group 1}
    \end{align*}
    These are also called wire clauses.\\

    Finally, we replace the lowest-priority group, Group 4, which previously stated
    that all gates should be correct, including those in the inactive circuit, 
    by wire clauses:
    \begin{align*}
        g_i^T \leftrightarrow g_{j,1}^T\\ 
        g_i^B \leftrightarrow g_{j,1}^B \tag{The new Group 4}
    \end{align*}
    
    \begin{lemma}\label{lemma-3sat-correct}
    Suppose $\alpha$ is a local maximum, i.e., cannot be improved by a 
    1-flip or 2-flip. Then all clauses in groups 0, 1, 3, and 4 are satisfied. 
    \end{lemma}

    \begin{proof}
    A clause in Group 0 can always be satisfied by changing $\alpha$ in one 
    position, namely the output of the incorrect gate. Now suppose 
    a clause in the new group 1 is violated, for concreteness 
    $(f \rightarrow (g_i^T \leftrightarrow g_{j,1}^T))$, so $f=1$ (TOP is active) 
    and $\alpha(g_i^T) \ne \alpha(g_{j,1}^T)$ (the wire from gate $g_i^T$ to 
    gate $g_j^T$ is inconsistent); we can satisfy this by flipping 
    $\alpha(g_{j,1}^T)$; this might violate gate $g_j^T$ and thus a 
    clause in Group 0, which has higher priority. However, 
    we are allowed to flip up to {\em two} variables: we flip 
    $g_{j,1}^T$ and, if need be, also $g_j^T$. This keeps Group 0 satisfied; 
    it might violate a wire clause for a wire leaving gate $j$; however,
    this is in Group 1 and has a lower priority due to the topological ordering
    of the circuit. 

    A clause in Group 3 requires that the output of the active 
    circuit  agree with the input of the inactive one; suppose one is violated, 
    for concreteness $(f \rightarrow (z_j^T \leftrightarrow x_j^B))$, 
    so $f=1$ (TOP is active). We will simply flip $x_j^B$. The 
    wires going from $x_j$ to the gates into which it feeds will now 
    be inconsistent; however, the corresponding wire clauses in Group 1 
    will be satisfied because $f=1$ (BOTTOM is inactive); some wire clauses 
    of Group 4 will become violated, but that's fine because they have 
    lower priority than the clause in Group 3. 

    Finally, if a clause in Group 4 is violated, it must be 
    in the inactive circuit (all wire clauses for the active circuit are 
    satisfied due to Group 1); for concreteness, BOTTOM is inactive 
    and $(g_i^B \leftrightarrow g_{j,1}^B)$ is violated. We flip 
    $g_{j,1}^B$ and, if need be, the output variable $g_j^B$. Clauses 
    in Group 0 are still satisfied because all gates are still correct; 
    clauses in Group 1 are satisfied because we only fiddled in BOTTOM, 
    which is inactive, so all its Group-1-clauses are satisfied anyway; 
    Group 3 stays satisfied because outputs of the inactive circuit BOTTOM are 
    not required to agree with the inputs of the active circuit TOP; 
    in Group 4 further wire clauses involving $g_{j^B}$ might become
    violated, but they have lower priority than the one we fixed.
    \end{proof}

    If Group 0, 1, 3, and 4 
    are satisfied, and the output of the inactive circuit is lexicographically 
    larger than the output of the active circuit, we can still flip $f$ 
    and one bit of $\vec{o}$ to improve $\alpha$, simply because the clauses
    of Group 2 are identical to the ones in the reduction to $4$-SAT/2-FLIP.
    This concludes the proof of \cref{theorem-lex-3-sat-pls-complete}

    Note that our reduction to $4$-SAT/2-FLIP uses 2-flips very sparingly, 
    only when we switch $f$ (switch which circuit is active). In the 
    reduction to $3$-SAT/2-FLIP we  use a 2-flip potentially every time we
    propagate a value through the circuit. 

    The shadow variables idea does not work here, because we cannot extend the clauses here as they are already 3-clauses.

\section{Cyclic \localmin}\label{section-cyclic-polm}

    We can strengthen the result from \cref{section-abelian-polm} to cyclic groups with a proof technique similar to the \NP{}-hardness proof for the global optima from \cite{SchederTantow2025}.

    \begin{theorem}\label{thm-cyclicpolm}
        \localmin{} is \PLS{}-complete even there exists a permutation $\sigma$ such that all permutations $\pi_i = \sigma^j$.
    \end{theorem}
    
    We reduce from $4$-SAT/1-FLIP with lexicographic payoffs, 
    which is \PLS-complete according to \cref{theorem-lex-4-sat-pls-complete}. Let $F$ be a
    $4$-CNF formula 
    $x_1, \dots, x_n$ its variables, and $C_1, \dots, C_m$ its clauses from 
    highest-priority to lowest-priority.

    Consider the first $n$ prime numbers $p_1 = 2, 3, 5, \dots, p_n$
    and set $N := p_1 \cdot p_2 \cdots p_n$. We let $\Z_k$ denote 
    the group $\left(\Z / k\Z, +\right)$ and also, by a slight abuse 
    of notation, the set
    $\{0, 1, \dots, k-1\}$. By the Chinese remainder theorem, 
    the function 
    \begin{align*}
        f: \Z_N & \rightarrow \Z_{p_1} \times \cdots \times \Z_{p_n} \\
        l   & \mapsto (l \textnormal{ mod } p_1, \dots, l \textnormal{ mod } p_n) 
    \end{align*}
    is bijective, and both the function and its inverse can be computed in polynomial 
    time. 

    \textbf{The positions.}\\
    Our positions will be arranged in {\em cycles}, one for each 
    variable and one for each clause. In each cycle, the positions are 
    numbered as $0, \dots, l-1$, where $l$ is the length of the cycle. 
    The initial bit string $s$ will place a $1$ on each $0$-position in each 
    cycle and $0$s elsewhere. Each permutation will rotate each 
    cycle by a given amount; thus, each element in the orbit of $s$ 
    will have exactly one $1$ on each cycle.

    \textbf{Variable positions.}\\
    For every variable $x_i$ we create a cycle of length $p_i$ 
    with positions $0, 1, \dots, p_i-1$. 
    An element in the orbit has one 1 per cycle and thus assigns each 
    variable $x_i$ 
    a value in $\{0, 1, \dots, p_i-1\}$. We call 
    the positions $2, \dots, p_i-1$ bad-color-positions and 
    make them very expensive (giving
    them the highest priority) to make sure that a local minimum only 
    assigns values $0$ and $1$ to the variables. Among bad-color-positions, 
    the order is arbitrary. 

    \textbf{Clause positions.}\\ Let $C$ be a clause of $F$ 
    and suppose it contains the variables $x_i, x_j, x_k$ and $x_l$. We create 
    a cycle of length $p_i p_j p_k p_l$. 
    Each $t \in \{0, \dots, p_ip_jp_kp_l-1\}$ becomes, when taking 
    remainders modulo $p_i, p_j, p_k, p_k$, a tuple $(r_i, r_j, r_k, r_l) \in 
    \Z_{p_i} \times \Z_{p_j}  \times \Z_{p_k} \times \Z_{p_l}$. Should this 
    triple be in $\{0,1\}^4$ and should it be the unique truth 
    assignment violating clause $C$, we call $t$ the {\em clause-violating position}
    on that cycle. 
    In our lexicographic ordering, we place the clause-violating 
    positions directly after the bad-color-positions; among themselves, we 
    order them as the ordering of the clauses of $F$ dictates: the 
    clause-violating position of clause $C_i$ has higher priority than 
    that of $C_{i+1}$. 

    \textbf{The permutations.}\\
    Let $\pi$ be the permutation that rotates each cycle by one (maps position $t$ on the cycle to position $t+1$, 
    and so on). For $m \in \N$ the permutation $\pi^m$ rotates the variable 
    cycle of $x_i$ by 
    $m$ mod $p_i$ and the clause cycle of a clause with variables 
    $x_i, x_j, x_k$ and $x_l$ by $m$ mod $p_i p_j p_k p_l$. Thus, the bit string 
    $v = s \circ \pi^l$ has a $1$ on a bad-color position if 
    $v$ does not represent a Boolean assignment to the variables; 
    if it does represent a Boolean assignment $\alpha$, 
    then $v$ has a $1$ on the clause-violating position of clause $C$ 
    if and only if $\alpha$ violates $C$. \\

    For each variable $x_i$ and each $r \in \Z_{p_i}$ let $m$ be such that 
    \begin{align*}
     m \equiv r \mod p_i \\
     m \equiv 0 \mod p_j \tag{for all $j \ne i$} 
    \end{align*}
    and set $\pi_{x_i, r} := \pi^m$.
    The idea is that the permutation $\pi_{x_i, r}$ changes the value of $x_i$ but 
    not of any other variable. The generators for our 
    \localmin{} instances are now all 
    $\pi_{x_i,r}$. 
    This gives a polynomially long list of permutations 
    $\pi^{e_1}, \pi^{e_2}, \dots, \pi^{e_M}$, all powers of a single permutation 
    $\pi$.
    \begin{lemma}
    Let $l$ be such that $s \circ \pi^l$ is a local minimum, i.e., 
    $s \circ \pi^{l} \circ \pi^{e_i} \geq_{\rm lex} s \circ \pi^{l}$ for 
    all $1 \leq i \leq M$. Then $s \circ \pi^{l}$ encodes 
    an assignment $\{x_1, \dots, x_n\} \rightarrow \{0,1\}$ that 
    is locally optimal for $F$ under 2-flips.
    \end{lemma}
    \begin{proof}
       Let $v := s \circ \pi^l$. Each variable cycle contains 
       exactly one position where $v$ has a 1. This assigns each $x_i$ a value 
        $b \in \{0,1, \dots, p_i-1\}$. Now if $b \not \in \{0,1\}$ then 
        $v$ has a 1 at position $b$ of this cycle---a bad-color-position. 
        Applying $\pi_{x_i, p_i-b}$ moves the 1 on this cycle to the $0$-position 
        but leaves all other variable cycles as they are. 
        We have moved a 1 away from a bad-color-position and thus 
        decreased the cost of $v$; hence $v$ is not a local minimum.

        Conversely, if $v$ is a local minimum, then $v$ has no $1$ at any 
        bad-color position and thus $v$ encodes, as described above, 
        a truth assignment $\alpha$ to the variables $x_1, \dots, x_n$. 
        Suppose $\alpha$ is not locally optimal for $F$ 
        but can be improved by flipping 
        a variable $x_i$. This flip can be realized 
        by the permutation $\pi_{x_i, r}$ where 
        $r$ is $1$ if we want to flip $x_i$ from $0$ to $1$ and 
        $p_i-1$ if we want to flip it from $1$ to $0$. Since flipping 
        $x_i$ is an improving step, 
        it additionally satisfies some clause $C$ and does not 
        additionally violate any higher priority clause $C'$; 
        thus, applying $\pi_{x_i, r}$ moves the 
        $1$ of $v$ away from the clause-violating position of $C$ but 
        places no new $1$ on any clause-violating position of any 
        higher-priority clause $C'$. Also, it does not move a 1 onto 
        any bad-color-position. Thus, $v$ is not a local minimum, either.

        We infer that every local minimum $v = s \circ \pi^l$ under
        the permutations $\pi^{e_1},\dots, \pi^{e_M}$ corresponds 
        to a local maximum of the $4$-SAT/1-FLIP instance $F$.
        This concludes the proof of \cref{thm-cyclicpolm}.
    \end{proof}

\section{Tightness}\label{section-tightness}

    The notion of {\em tightness} was introduced by Schäffer and Yannakakis in \cite{DBLP:journals/siamcomp/SchafferY91} to also reason about the complexity of the standard algorithm. The standard algorithm uses a starting solution and always follows the best neighbor until it finds a local optimum. Finding the solution of the standard algorithm for \flip{} is \PSPACE{}-complete~\cite{DBLP:conf/stoc/PapadimitriouSY90}. Tight reductions 
    preserve the \PSPACE{}-completeness of the standard algorithm.
    
    For this, we consider the transition graph $TG(I)$. Its vertices 
    are the solutions to instance $I$, and each feasible solution $x$ has 
    a directed edge to each of its improving neighbors in $N(x)$. A \PLS-reduction $(f,g)$ from $P$ to $Q$ is called \emph{tight} if for every instance $I$ of $P$ there exists a set $\mathcal{R}$ of feasible solutions for $f(I)$ such that
    \begin{enumerate}
        \item $\mathcal{R}$ contains all local optima of $f(I)$
        \item For every solution $s$ of $I$ it is possible to construct in polynomial time a feasible solution $t \in \mathcal{R}$ such that $g(I,t) = s$
        \item If the transition graph of $TG(f(I))$ contains a path from $q$ to $q'$ such that both $q$ and $q'$ are in $\mathcal{R}$ and all other intermediate nodes are not in $\mathcal{R}$, let $p = g(I, q)$ and $p' = g(I, q')$ be the corresponding solutions in $P$. Then either $p = p'$ or there is an arc from $p$ to $p'$ in $TG(I)$.
    \end{enumerate}

    We now show that the reductions are tight. We start with the reductions from \flip{} to \lexmaxksat{4}.
    This reduction uses many steps (often called a super step) to simulate a single step in the \flip{} instance. We have to make sure that there is no earlier switch in a super step. 
    
    \begin{proposition}\label{prop-tight-4-sat}
        The reduction from \flip{} to \lexmaxksat{4} in \cref{section-lex-4-sat} is tight.
    \end{proposition}
    \begin{proof}
        Let $C$ be the instance of \flip{} and $F$ be the formula obtained by our reduction. 
        We define $\mathcal{R}$ to be the set of assignments to $F$ where all clauses of group 0, 1 and 1b are fulfilled. 
        If this is true, then we know that the active circuit is correct and the Hamming distance between the inputs of both circuits is at most 1. By \cref{lem-4sat-circuit-correct} and \cref{lem-4sat-input-correct} we know that any local optimum is in $\mathcal{R}$. 
        This shows that Point 1 holds.
    
        We can take any solution $\mathbf{x}$ of \flip{}, set the input $\mathbf{x}^T$ of TOP  to $\mathbf{x}$ 
        and set all variables in the circuit such that all gates are correct and $\mathbf{z}^T = \mathbf{x}^B$ and $f=1$.
        The resulting assignment $\alpha$ satisfies $g(C,\alpha) = \mathbf{x}$. This shows that Point 2 holds.

        Let us now show Point 3. First observe that if $\alpha \in \mathcal{R}$ and there is an improving 
        step $(\alpha, \beta)$ in $TG(F)$, then $\beta$ also satisfies groups 0, 1b, and 1, since they 
        stay fulfilled along any improving path. So $\beta \in \mathcal{R}$, too. Now 
        let $q, q'\in \mathcal{R}$ as in Point 3: there is a path in $TG(F)$ from $q$ to $q'$ and no other vertex on the path 
        is in $\mathcal{R}$. We have seen that {\em all} vertices on an improving path starting in $\mathcal{R}$ are in $\mathcal{R}$, 
        so actually $(q,q')$ is an arc in $TG(F)$, so $q = \alpha$ and $q' = \beta$, an assignment obtained from $\alpha$ 
        by changing one variable. Without loss of generality, assume that TOP is active.
        
        There are two possibilities. If $\alpha$ and $\beta$ differ in a shadow variable, global output variable, 
        or a variable in the inactive circuit, then $\alpha$ and $\beta$ represent the same solution to $C$, namely 
        $\alpha(\mathbf{x}^T)$.
        The other possibility is that it was created by flipping $f$. This only occurs when the inactive circuit is correct 
        and its  output is better than the output of the active circuit. 
        But then $g(C, \beta) = \alpha(\mathbf{x}^B)$ is an improving \flip{}-neighbor of 
        $g(C,\alpha) = \alpha(\mathbf{x}^T)$, which means that this is an arc in $TG(I)$.
    \end{proof}

    In fact, the CNF formula resulting from our reduction can be designed to be 
    {\em sensitive}, meaning that for every assignment $\alpha$ and any variable $v$ 
    there exists a clause $C$ such that changing the value of $v$ changes 
    the truth value of $C$:
    $\alpha(C) \ne \beta(C)$ for $\beta := \alpha[v \mapsto \neg \alpha(v)]$.

    \begin{proposition}\label{prop-sensitive-4-sat}
        Sensitive \lexmaxksat{4} is \PLS{}-complete under a tight reduction.
    \end{proposition}

    \begin{proof}
        Note that our reduction from \cref{subsection-lex-single-flip} fulfills 
        the sensitivity condition almost all variables: the global output variables change clauses in group 2; shadow variables in group 6; gate variables in group 4; and input variables of the inactive circuit in group 3.

        The only remaining variables are $f$ and input variables in the active circuit. It might be that changing one input does not change the correctness of any clause
        stemming from a NAND-gate. In order to counter that, we 
        can route each input through an ``identity gate'' 
        behind each input (and simulate it by 
        two NAND-gates). This does not change the circuit semantics, but changing the input now either fulfills or violates a gate clause.

        For $f$ we can add a single clause that fixes $f$ to some value. This clause has the lowest weight of all. For the tightness this however this leads to problems as in the case of neighbors with equal outputs as the input of the top circuit is preferred to the input of the bottom circuit for equal outputs. We can however ensure that this scenario neither comes up by assuming that the circuit creates different outputs for different inputs e.g. by extending the output with the input in the lowest positions. 

        With these modifications, the proof works analogously to the proof of \cref{prop-tight-4-sat}.
    \end{proof}

    Reducing from sensitive \lexmaxksat{4} allows us now to make all remaining reductions 
    tight.
    
    \begin{proposition}
        The reduction from sensitive \lexmaxksat{4} to Abelian \polm{} and to finding approximate Nash equilibria is tight.
    \end{proposition}
    \begin{proof}
        We first see that the transition graphs of the instances obtained by the reduction contains the transition graph of \lexmaxksat{4}, since every improving flip of \lexmaxksat{4} is directly represented as a single flip in the new instances. Hence, we can take all feasible solutions as $\mathcal{R}$ and fulfill the first two points.

        For the third point, we have to make sure that there are no internal improving flips that change the assignment while the fulfilled clauses stay the same. 
        Because if there were, then this would not be an improving move in the \lexmaxksat{4} instance but might be in the newly created instances since it changes the delay on some more important resource from 3 to 2 while the delay on some less important resource increases from 1 to 2. Similarly, in the \polm{} instance this would mean that some very important fulfilling position by the order has a one and after applying it we map the one to some less important fulfilling position.

        By \cref{prop-sensitive-4-sat} we know however that these internal improving flips can never occur and any flip changes the truth value of at least some clause. Thus, any improving flip in the $f(I)$ instance here is an improving flip in the original
        instance $I$.
    \end{proof}

    For cyclic \polm{} from the previous section the same proof idea works, but as $\mathcal{R}$ we choose the permutations that define an assignment, i.e. there are no ones at bad-color positions.

\section{FP-completeness for \lexmaxksat{2}}

    In contrast to the decision variant 2-SAT, which is \NL{}-complete, 
    and the arbitrarily-weighted \localmaxksat{2}, which is \PLS{}-complete\cite{DBLP:journals/siamcomp/SchafferY91}, we show that 
    \localmaxksat{2} with lexicographic costs (i.e. \lexmaxksat{2}) is \FP{}-complete. 
    We have already seen a polynomial time algorithm in the introduction; 
    now we will show hardness.
    
    \begin{proposition}
        \lexmaxksat{2} is \FP-hard.
    \end{proposition}
    \begin{proof}
        We reduce from the circuit evaluation problem. We are given a circuit $C$ 
        and an input $\mathbf{x} \in \{0,1\}^n$. The gates of $C$ are given
        in topological order. The goal is to compute $C(\mathbf{x}) \in \{0,1\}^m$, 
        i.e., the output of the circuit.         
        We can assume that the circuit consists only of NAND-gates. We use a variable $x_i$ for each input and 
        a variable $g_j$ for each gate.
    
        The highest clauses fix the input variables to the current input. Afterwards, we use the following three clauses for each gate. 
        If gate $g_j$ has inputs $y_1$ and $y_2$, we produce the following 
        block containing three clauses:
        \begin{enumerate}
            \item $\neg y_1 \rightarrow g_j$
            \item $\neg y_2 \rightarrow g_j$
            \item $\neg g_j$
        \end{enumerate}
        
        The blocks are ordered according to the topological ordering of the 
        circuit; within each block they are ordered as shown above. \\

        We claim that in any local optimum $\alpha$, it holds that  $g_j = {\rm NAND}(y_1, y_2)$ holds. Suppose that this is not the case.\\
        
        \textbf{Case 1:} $\alpha(g_j)=1$ and  $\alpha(y_1) = \alpha(y_1) = 1$.
        Then the third clause above, $\neg g_j$, is violated, and setting $g_j$ to $0$ 
        will satisfy all three clauses above. \\

        \textbf{Case 0:} $\alpha(g_j)=0$ and at least one of $y_1, y_2$ is mapped to 
        $0$. Then one of the first two clauses must is violated; flipping $g$ 
        satisfies it; it will violate the third clause, but it is still an improving 
        flip.\\

        So we conclude that  under $\alpha$ every gate is correctly evaluated. 
        Therefore, the value of the $m$ output gate variables under $\alpha$ 
        give us the correct output of the 
        circuit: $C(\mathbf{x}$).
    \end{proof}
\end{document}